%% file: main.tex
\documentclass[10pt]{article}

\usepackage[margin=2.5cm, includefoot, footskip=30pt]{geometry}
\usepackage{tikz}
\usepackage{amsmath}
\usepackage{amssymb}
\usepackage{mathtools}
\usepackage{amsthm}
\usepackage{graphicx}
\usepackage{subcaption}
\usepackage{standalone}
\usepackage{booktabs}
\usepackage{setspace}
\usepackage[algoruled,lined]{algorithm2e}
\usepackage[noend]{algpseudocode}
\usepackage{wrapfig}
\usepackage{hyperref}
\usepackage{authblk}
\usepackage[toc,page]{appendix}
\usetikzlibrary{calc, shapes, patterns, decorations.pathreplacing}

\makeatletter
\def\BState{\State\hskip-\ALG@thistlm}
\makeatother

\newcommand{\R}{\mathbb{R}}
\newtheorem{theorem}{Theorem}
\usetikzlibrary{decorations.pathmorphing, decorations.pathreplacing, angles,
                quotes, calc, er, positioning}

\newtheorem{lemma}[theorem]{Lemma}

\title{Using a theory of mind to find best responses to memory-one strategies.}
\author[1, 2, *]{Nikoleta E. Glynatsi}
\author[1]{Vincent A. Knight}

\affil[1]{Cardiff University, School of Mathematics, Cardiff, United Kingdom}
\affil[2]{Max Planck Institute for Evolutionary Biology, Pl\"{o}n, Germany}
\affil[*]{Corresponding author: Nikoleta E. Glynatsi, glynatsi@evolbio.mpg.de}
\date{}
\begin{document}

\maketitle

\newpage

\begin{abstract}
    Memory-one strategies are a set of Iterated Prisoner's Dilemma strategies
    that have been praised for their mathematical tractability and performance
    against single opponents. This manuscript investigates
    \textit{best response} memory-one strategies with a theory of mind for
    their opponents. The results add to the literature that has shown that
    extortionate play is not always optimal by showing that optimal play is
    often not extortionate.
    They
    also provide evidence that memory-one strategies suffer from their limited
    memory in multi agent interactions and can be out performed by
    optimised strategies with longer memory.
    We have developed a theory that has allowed to explore the entire
    space of memory-one strategies. The framework presented is suitable to
    study memory-one strategies in the Prisoner's Dilemma, but also
    in evolutionary processes such as the Moran process,
    Furthermore, results on the stability of defection in populations of
    memory-one strategies are also obtained.
\end{abstract}

The Prisoner's Dilemma (PD) is a two player game used in understanding the
evolution of cooperative behaviour, formally introduced in~\cite{Flood1958}.
Each player has two options, to cooperate (C) or to defect (D). The decisions
are made simultaneously and independently. The normal form representation of the
game is given by:

\begin{equation}\label{equ:pd_definition}
    S_p =
    \begin{pmatrix}
        R & S  \\
        T & P
    \end{pmatrix}
    \quad
    S_q =
    \begin{pmatrix}
        R & T  \\
        S & P
    \end{pmatrix}
\end{equation}

where \(S_p\) represents the utilities of the row player and \(S_q\) the
utilities of the column player. The payoffs, \((R, P, S, T)\), are constrained
by \(T > R > P > S\) and \(2R > T + S\), and the most common values used in the
literature are \((R, P, S, T) = (3, 1, 0, 5)\)~\cite{Axelrod1981}.
The numerical experiments of our manuscript are carried out using these
payoff values.
The PD is a one shot game, however, it is commonly studied in a manner where the
history of the interactions matters. The repeated form of the game is called the
Iterated Prisoner's Dilemma (IPD).

Memory-one strategies are a set of IPD strategies that have been
studied thoroughly in the literature~\cite{Nowak1990, Nowak1993}, however, they have gained
most of their attention when a certain subset of memory-one strategies was
introduced in~\cite{Press2012}, the zero-determinant strategies (ZDs). In~\cite{Stewart2012} it
was stated that ``Press and Dyson have fundamentally changed the viewpoint on
the Prisoner's Dilemma''.
A special case of ZDs are extortionate strategies that choose their actions so that a linear relationship is forced
between the players' score ensuring that they will always
receive at least as much as their opponents. ZDs are
indeed mathematically unique and are proven to be robust in pairwise
interactions, however, their true effectiveness in tournaments and
evolutionary dynamics has been questioned~\cite{adami2013, Hilbe2013b,
Hilbe2013, Hilbe2015, Knight2018, Harper2015}.

The purpose of this work is to reinforce the literature on the limitations of
extortionate strategies by considering a new approach. More specifically, by
considering best response memory-one strategies with a theory of mind of
their opponents. There are several works in the literature that have considered
strategies with a theory of mind~\cite{Han2011, De2013, Devaine2014, Han2012,
Press2012, Stewart2012}. These works defined ``theory of mind'' as intention
recognition~\cite{Han2011, Han2012, De2013, Devaine2014} and as the ability of a
strategy to realise that their actions can influence
opponents~\cite{Stewart2012}. Compared to these works, theory of mind is defined
here as the ability of a strategy to know the behaviour of their opponents
and alter their own behaviour in response to that.

We present a closed form algebraic expression for the utility of a
memory-one strategy against a given set of opponents and a compact method of
identifying it's best response to that given set of opponents.
The aim is to evaluate whether a best response memory-one
strategy behaves in a zero-determinant way which in turn indicates whether it
can be extortionate. We do this using a linear algebraic approach presented
in~\cite{Knight2019}. This is done in tournaments with two opponents.
Moreover, we introduce a framework that allows the comparison of
an optimal memory-one strategy and an optimised strategy which has a larger
memory.

To illustrate the analytical results obtained in this manuscript a number of
numerical experiments are run. The source code for these experiments has been
written in a sustainable manner~\cite{Benureau2018}.
It is open source (\url{https://github.com/Nikoleta-v3/Memory-size-in-the-prisoners-dilemma})
and tested which ensures the validity of the results. It has also been archived
and can be found at~\cite{nikoleta_glynatsi_2019}.

\section{Methods}

One specific advantage of memory-one strategies is their mathematical
tractability. They can be represented completely as an element of \(\R^{4}_{[0, 1]}\). This
originates from~\cite{Nowak1989} where it is stated that if a strategy is
concerned with only the outcome of a single turn then there are four possible
`states' the strategy could be in; both players cooperated (\(CC\)),
the first player cooperated whilst the second player defected (\(CD\)),
the first player defected whilst the second player cooperated (\(DC\)) and
both players defected (\(DD\)).
Therefore, a memory-one strategy can be denoted by the probability vector of
cooperating after each of these states; \(p=(p_1, p_2, p_3, p_4) \in \R_{[0,1]}
^ 4\).

In~\cite{Nowak1989} it was shown that it is not necessary to simulate the play
of a strategy $p$ against a memory-one opponent $q$. Rather this exact behaviour
can be modeled as a stochastic process, and more specifically as a Markov chain
whose corresponding transition matrix \(M\) is
given by Eq.~\ref{eq:transition_matrix}. The long run steady state probability
vector \(v\), which is the solution to \(v M = v\), can be
combined with the payoff matrices of Eq.~\ref{equ:pd_definition} to give the expected
payoffs for each player. More specifically, the utility for a memory-one
strategy \(p\) against an opponent \(q\), denoted as \(u_q(p)\), is given by
Eq.~\ref{eq:press_dyson_utility}.

\begin{equation}\label{eq:transition_matrix}
    \resizebox{.5\hsize}{!}{$\input{tex/m_matrix.tex}$}
\end{equation}

\begin{equation}\label{eq:press_dyson_utility}
    u_q(p) = v \cdot (R, S, T, P).
\end{equation}

This manuscript has explored the form of \(u_q(p)\), to the best of the authors knowledge no
previous work has done this, and Theorem~\ref{theorem_one} states that \(u_q(p)\) is given by a ratio
of two quadratic forms~\cite{kepner2011}.

\begin{theorem}\label{theorem_one}
    The expected utility of a memory-one strategy \(p\in\mathbb{R}_{[0,1]}^4\)
    against a memory-one opponent \(q\in\mathbb{R}_{[0,1]}^4\), denoted
    as \(u_q(p)\), can be written as a ratio of two quadratic forms:

    \begin{equation}\label{eq:optimisation_quadratic}
    u_q(p) = \frac{\frac{1}{2}pQp^T + cp + a}
                {\frac{1}{2}p\bar{Q}p^T + \bar{c}p + \bar{a}},
    \end{equation}
    where \(Q, \bar{Q}\) \(\in \R^{4\times4}\) are square matrices defined by the
    transition probabilities of the opponent \(q_1, q_2, q_3, q_4\) as follows:

    \begin{center}
    \begin{equation}
    \resizebox{0.9\linewidth}{!}{\arraycolsep=2.5pt%
    \boldmath\(
    Q = \input{tex/q_numerator}\)},
    \end{equation}
    \begin{equation}\label{eq:q_bar_matrix}
    \resizebox{0.8\linewidth}{!}{\arraycolsep=2.5pt%
    \boldmath\(
    \bar{Q} =  \input{tex/q_denominator}\)}.
    \end{equation}
    \end{center}

    \(c \text{ and } \bar{c}\) \(\in \R^{4 \times 1}\) are similarly defined by:

    \begin{equation}\label{eq:q_matrix_numerator}
    \resizebox{0.6\linewidth}{!}{\arraycolsep=2.5pt%
    \boldmath\(c = \input{tex/c_numerator}\),}
    \end{equation}
    \begin{equation}\label{eq:q_matrix_denominator}
    \resizebox{0.3\linewidth}{!}{\arraycolsep=2.5pt%
    \boldmath\(\bar{c} = \input{tex/c_denominator}\),
    }
    \end{equation}
    and the constant terms \(a, \bar{a}\) are defined as \(a = \input{tex/numerator_constant}\) and
    \(\bar{a} = \input{tex/denominator_constant}\).
\end{theorem}

The proof of Theorem~\ref{theorem_one} is given in
Appendix~\ref{appendix:theorem_one}.
Theorem~\ref{theorem_one} can be extended to consider multiple
opponents. The IPD is commonly studied in tournaments and/or Moran Processes
where a strategy interacts with a number of opponents. The payoff of a player in
such interactions is given by the average payoff the player received against
each opponent. More specifically the expected utility of a memory-one strategy
against \(N\) opponents is given by:

\begin{align}\label{eq:tournament_utility}
       & \frac{1}{N} \sum\limits_{i=1} ^ {N} {u_q}^{(i)} (p) =
       \frac{\frac{1}{N} \sum\limits_{i=1} ^ {N} (\frac{1}{2} pQ^{(i)} p^T + c^{(i)} p + a^ {(i)})
       \prod\limits_{\tiny\begin{array}{l} j=1 \\ j \neq i \end{array}} ^
       N (\frac{1}{2} p\bar{Q}^{(j)} p^T + \bar{c}^{(j)} p + \bar{a}^ {(j)})}
       {\prod\limits_{i=1} ^ N (\frac{1}{2} p\bar{Q}^{(i)} p^T + \bar{c}^{(i)} p + \bar{a}^ {(i)})}.
\end{align}

Eq.~\ref{eq:tournament_utility} is the average score (using Eq.~\ref{eq:optimisation_quadratic}) against the set of opponents.

Estimating the utility of a memory-one strategy against any number of opponents
without simulating the interactions is the main result used in the rest of this manuscript.
It will be used to obtain best response memory-one strategies in tournaments
in order to explore the limitations of extortion
and restricted memory.

\section{Results}\label{section:results}

The formulation as presented in Theorem~\ref{theorem_one} can be used to
define \textit{memory-one best response} strategies as a multi dimensional
optimisation problem given by:

\begin{equation}\label{eq:mo_tournament_optimisation}
    \begin{aligned}
    \max_p: & \ \sum_{i=1} ^ {N} {u_q}^{(i)} (p)
    \\
    \text{such that}: & \ p \in \R_{[0, 1]}
    \end{aligned}
\end{equation}

Optimising this particular ratio of quadratic forms is not trivial. It can be
verified empirically for the case of a single opponent that there exists at
least one point for which the definition of concavity does not hold.
The non concavity of \(u(p)\) indicates multiple local
optimal points. This is also intuitive. The best response against a cooperator,
\(q=(1, 1, 1, 1)\), is a defector \(p^*=(0, 0, 0, 0)\). The strategies
\(p=(\frac{1}{2}, 0, 0, 0)\) and \(p=(\frac{1}{2}, 0, 0, \frac{1}{2})\) are also
best responses. The approach taken here is to introduce a compact way of
constructing the discrete candidate set of all local optimal points, and evaluating
the objective function Eq.~\ref{eq:tournament_utility}. This gives the best
response memory-one strategy. The approach is given in
Theorem~\ref{memone_group_best_response}.

\begin{theorem}\label{memone_group_best_response}

    The optimal behaviour of a memory-one strategy player \(p^* \in \R_{[0, 1]} ^
    4\) against a set of \(N\) opponents \(\{q^{(1)}, q^{(2)}, \dots, q^{(N)} \}\)
    for \(q^{(i)} \in \R_{[0, 1]} ^ 4\) is given by:

    \[p^* = \textnormal{argmax}\sum\limits_{i=1} ^ N  u_q(p), \ p \in S_q.\]

    The set \(S_q\) is defined as all the possible combinations of:

    {\scriptsize
    \begin{equation}\label{eq:s_q_set}
        S_q =
        \left\{p \in \mathbb{R} ^ 4 \left|
            \begin{aligned}
                \bullet\quad p_j \in \{0, 1\} & \quad \text{and} \quad \frac{d}{dp_k}
                \sum\limits_{i=1} ^ N  u_q^{(i)}(p) = 0 \\
                & \quad \text{for all} \quad j \in J \quad \&  \quad k \in K  \quad \text{for all} \quad J, K \\
                & \quad \text{where} \quad J \cap K = \O \quad
                \text{and} \quad J \cup K = \{1, 2, 3, 4\}.\\
                \bullet\quad  p \in \{0, 1\} ^ 4
            \end{aligned}\right.
        \right\}.
    \end{equation}
    }

    Note that there is no immediate way to find the zeros of \(\frac{d}{dp}
    \sum\limits_{i=1} ^ N  u_q(p)\) where,

    {\scriptsize
    \begin{align}\label{eq:mo_tournament_derivative}
        \frac{d}{dp} \sum\limits_{i=1} ^ {N} {u_q}^{(i)} (p) & = \displaystyle\sum\limits_{i=1} ^ {N}
        \frac{\left(pQ^{(i)} + c^{(i)}\right) \left(\frac{1}{2} p\bar{Q}^{(i)} p^T + \bar{c}^{(i)} p + \bar{a}^ {(i)}\right)}
        {\left(\frac{1}{2} p\bar{Q}^{(i)} p^T + \bar{c}^{(i)} p + \bar{a}^ {(i)}\right)^ 2}
        - \frac{\left(p\bar{Q}^{(i)} + \bar{c}^{(i)}\right) \left(\frac{1}{2} pQ^{(i)} p^T + c^{(i)} p + a^ {(i)}\right)}
        {\left(\frac{1}{2} p\bar{Q}^{(i)} p^T + \bar{c}^{(i)} p + \bar{a}^ {(i)}\right)^ 2}
    \end{align}
    }

    For \(\frac{d}{dp} \sum\limits_{i=1} ^ N  u_q(p)\) to equal zero then:

    {\scriptsize
    \begin{align}\label{eq:polynomials_roots}
        \displaystyle\sum\limits_{i=1} ^ {N}
        \left(pQ^{(i)} + c^{(i)}\right) \left(\frac{1}{2} p\bar{Q}^{(i)} p^T + \bar{c}^{(i)} p + \bar{a}^ {(i)}\right)
        - \left(p\bar{Q}^{(i)} + \bar{c}^{(i)}\right) \left(\frac{1}{2} pQ^{(i)} p^T + c^{(i)} p + a^ {(i)}\right)
        & = 0, \quad {while} \\
        \displaystyle\sum\limits_{i=1} ^ {N} \frac{1}{2} p\bar{Q}^{(i)} p^T + \bar{c}^{(i)} p + \bar{a}^ {(i)} & \neq 0.
    \end{align}}

\end{theorem}

The proof of Theorem~\ref{memone_group_best_response} is given in
Appendix~\ref{appendix:memone_group_best_response}.
Finding best response memory-one strategies is analytically feasible using the
formulation of Theorem~\ref{memone_group_best_response} and resultant
theory~\cite{Jonsson2005}. However, for large systems building the resultant
becomes intractable. As a result, best responses will be estimated
heuristically using a numerical method, suitable for problems with local optima,
called Bayesian optimisation~\cite{Mokus1978}.

\subsection{Limitations of extortionate behaviour}

In multi opponent settings,
where the payoffs matter, strategies trying to exploit their opponents will
suffer.
Compared to ZDs, best response memory-one strategies, which have a
theory of mind of their opponents, utilise their behaviour in order to gain the
most from their interactions. The question that arises then is whether best
response strategies are optimal because they behave in an extortionate way.

The results of this section use Bayesian optimisation to generate a data set of best response
memory-one strategies for \(N=2\) opponents.
The data set is available at~\cite{glynatsi2019}. It contains a total of 1000 trials
corresponding to 1000 different instances of a best response strategy in
tournaments with \(N=2\). For each trial a set of 2 opponents is
randomly generated and the memory-one best response against them is found.
In order to investigate whether best responses
behave in an extortionate matter the SSE method described in ~\cite{Knight2019} is used. More
specifically,
in~\cite{Knight2019} the point \(x^*\), in the space of memory-one strategies,
that is
the nearest extortionate strategy to a given strategy \(p\) is
given by,

\begin{equation}\label{eqn:x_star_formula}
    x^* = {\left(C^{T}C\right)}^{-1}C^{T}\bar{p}
\end{equation}

where \(\bar{p}=(p_1 - 1, p_2 - 1, p_3, p_4)\) and

\begin{equation}\label{eq:definition_of_C}
    C =
    \begin{bmatrix}
        R - P & R- P \\
        S - P & T- P \\
        T - P & S- P \\
        0     & 0 \\
    \end{bmatrix}.
\end{equation}

Once this closest ZDs is found, the squared norm of the remaining error is referred to as sum of squared errors
of prediction (SSE):

\begin{equation}\label{eqn:x_SSError_formula}
    \text{SSE} = {\bar{p}} ^ T \bar{p} -
           \bar{p} C \left(C ^ T C \right) ^ {-1} C ^ T \bar{p}
         = {\bar{p}} ^ T \bar{p} - \bar{p} C x ^ *
\end{equation}

Thus, SSE is defined as how far a strategy is from behaving as a ZD. A
high SSE implies a non extortionate behaviour.
The distribution of SSE for the best response in tournaments
(\(N=2\)) is given in
Figure~\ref{fig:sse_distributions}. Moreover, a statistical summary of the SSE
distribution is given in Table~\ref{table:sserror_stats}.

\begin{table}[!htbp]
    \begin{center}
    \resizebox{.65\columnwidth}{!}{%
    \begin{tabular}{lrrrrrrrrrrr}
    \toprule
    mean & std  & 5\% & 50\% &  95\% & max & median & skew & kurt\\
    \midrule
    0.34  & 0.40  & 0.028  & 0.17  & 1.05  & 2.47  & 0.17  & 1.87 & 3.60 \\
    \end{tabular}}
    \end{center}
    \caption{SSE of best response memory-one when \(N=2\)}\label{table:sserror_stats}
\end{table}

\begin{figure}[!htbp]
    \begin{center}
        \includegraphics[width=.5\linewidth]{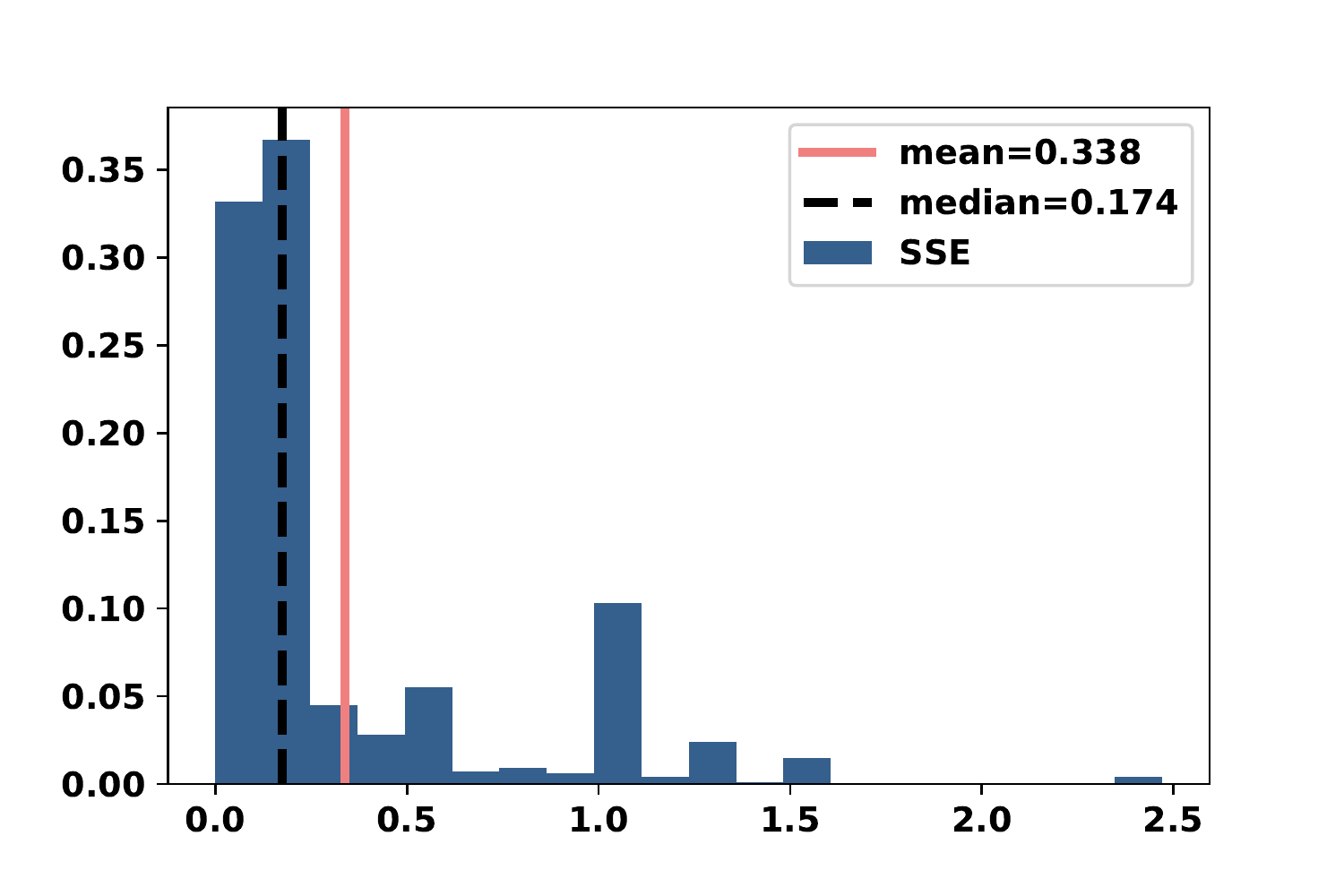}
    \end{center}
    \caption{SEE distribution for best response in tournaments with \(N=2\).}\label{fig:sse_distributions}
\end{figure}

For the best response in tournaments with \(N=2\) the
distribution of SSE is skewed to the left, indicating that the best response
does exhibit ZDs behaviour and so could be extortionate, however, the best
response is not uniformly a ZDs. A positive measure of skewness and kurtosis,
and a mean of 0.34 indicate a heavy tail to the right. Therefore, in several
cases the strategy is not trying to extort its opponents. This highlights the importance of adaptability since the best response strategy against an opponent is rarely (if ever) a unique ZDs.

\subsection{Limitations of memory size}

The other main finding presented in~\cite{Press2012} was that
short memory of the strategies was all that was needed.
We argue that the second limitation of ZDs in multi opponent
interactions is that of their restricted memory.
To demonstrate the effectiveness of memory in the IPD we explore a best response
longer-memory strategy against a given set of memory-one opponents,  and compare
it's performance to that of a memory-one best response.

In~\cite{Harper2017}, a strategy called \textit{Gambler} which makes
probabilistic decisions based on the opponent's \(n_1\) first moves, the
opponent's \(m_1\) last moves and the player's \(m_2\) last moves was
introduced. In this manuscript Gambler with parameters: $n_1 = 2, m_1 = 1$ and $m_2 = 1$ is used
as a longer-memory strategy.
By considering the opponent's first two moves, the opponents last move and the
player's last move, there are only 16 $(4 \times 2 \times 2)$ possible outcomes
that can occur, furthermore, Gambler also makes a probabilistic decision of
cooperating in the opening move. Thus, Gambler is a function \(f: \{\text{C,
D}\} \rightarrow [0, 1]_{\R}\). This can be hard coded as an element
of \([0, 1]_{\R} ^ {16 + 1}\), one probability for each outcome plus the opening
move. Hence, compared to Eq.~\ref{eq:mo_tournament_optimisation}, finding an
optimal Gambler is a 17 dimensional problem given by:

\begin{equation}\label{eq:gambler_optimisation}
    \begin{aligned}
    \max_p: & \ \sum_{i=1} ^ {N} {U_q}^{(i)} (f)
    \\
    \text{such that}: & \ f \in \R_{[0, 1]}^{17}
    \end{aligned}
\end{equation}

Note that Eq. \ref{eq:tournament_utility} can not be used here for the utility
of Gambler, and actual simulated players are used. This is done using~\cite{axelrodproject}
with 500 turns and 200 repetitions, moreover, Eq. \ref{eq:gambler_optimisation}
is solved numerically using Bayesian optimisation.

Similarly to the previous section, a large data set has been generated with
instances of an optimal Gambler and a memory-one best response, available
at~\cite{glynatsi2019}. Estimating a best response Gambler (17 dimensions) is
computational more expensive compared to a best response memory-one (4
dimensions). As a result, the analysis of this section is based on a total of
152 trials. As before, for each trial \(N=2\) random opponents have been selected.

The ratio between Gambler's utility and the best response memory-one strategy's utility has been calculated and its distribution in
given in Fig.~\ref{fig:utilities_gambler_mem_one}.
It is evident from Fig.~\ref{fig:utilities_gambler_mem_one} that
Gambler always performs as well as the best response memory-one strategy and often performs better. There are
no points where the ratio value is less than 1, thus Gambler never performed less
than the best response memory-one strategy and in places outperforms it.
However, against two memory-one opponents Gambler's performance is better than
the optimal memory-one strategy. This is evidence that in the case of multiple
opponents, having a
shorter memory is limiting.

\begin{figure}[!htbp]
    \centering
    \includegraphics[width=.5\textwidth]{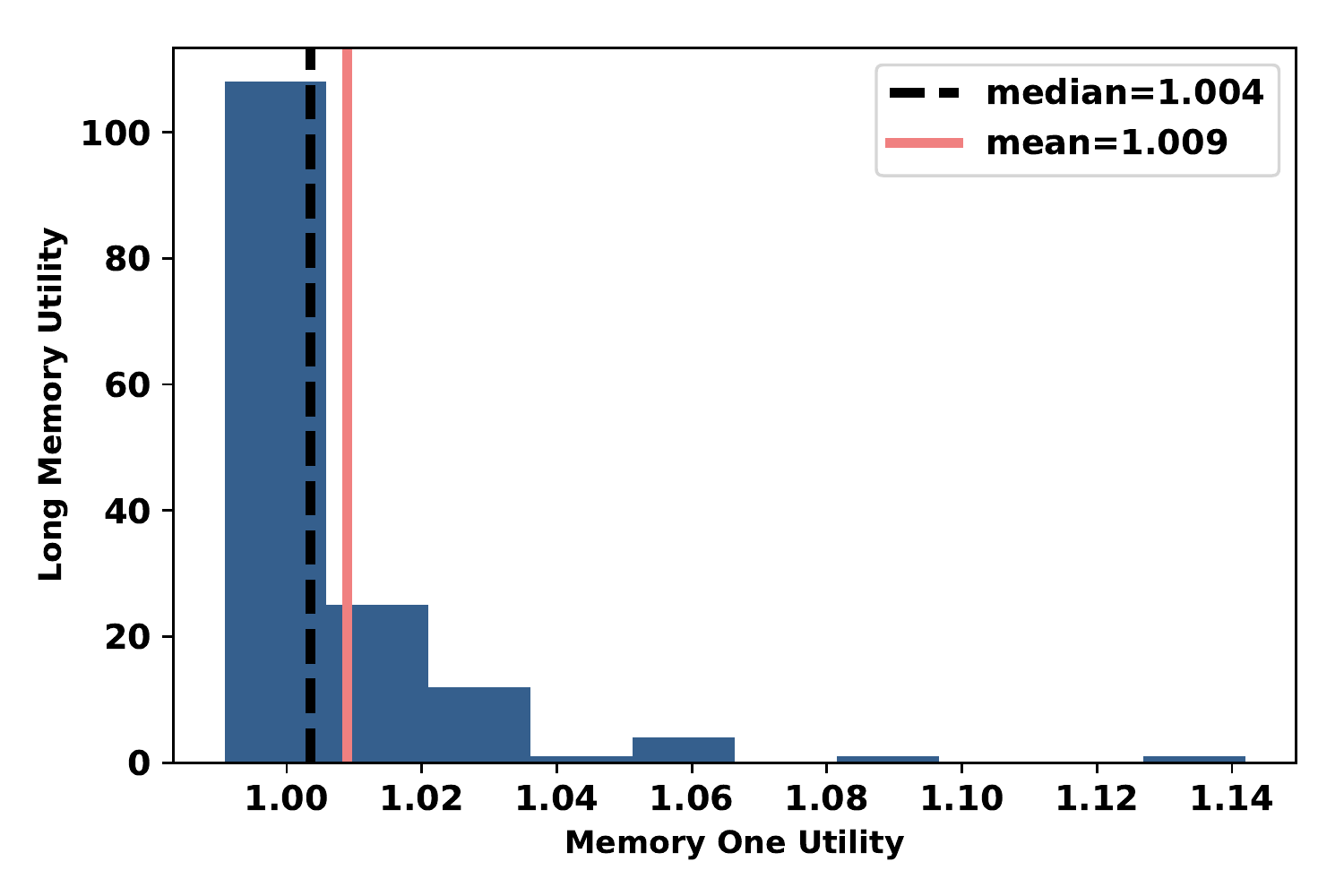}
    \caption{The ratio between the utilities of Gambler and best response memory-one
    strategy for 152 different pair of opponents.}\label{fig:utilities_gambler_mem_one}
\end{figure}

\subsection{Dynamic best response player}

In several evolutionary settings
such as Moran Processes self interactions are key. Previous work has
identified interesting results such as the appearance of self recognition
mechanisms when training strategies using evolutionary algorithms in Moran
processes~\cite{Knight2018}. This aspect of reinforcement learning can be done for
best response memory-one strategies, as presented in this manuscript, by incorporating the strategy itself in the
objective function as shown in Eq.~\ref{eq:mo_tournament_optimisation}.
Where \(K\) is the number of self interactions that will take place.

\begin{equation}\label{eq:mo_evolutionary_optimisation}
\begin{aligned}
\max_p: & \ \frac{1}{N} \sum\limits_{i=1} ^ {N} {u_q}^{(i)} (p) + Ku_p(p)
\\
\text{such that}: & \ p \in \R_{[0, 1]}
\end{aligned}
\end{equation}

For determining the memory-one best response with self interactions, an
algorithmic approach called \textit{best response dynamics} is proposed. The
best response dynamics approach used in this manuscript is given by
Algorithm~\ref{algo:best_response_dynamics}.

\begin{center}
\begin{minipage}{.55\textwidth}
\begin{algorithm}[H]
       $p^{(t)}\leftarrow (1, 1, 1, 1)$\;
       \While{$p^{(t)} \neq p ^{(t -1)}$}{
       $p^{(t + 1)} =  \text{argmax} \frac{1}{N} \sum\limits_{i=1} ^ {N} {u_q}^{(i)}
       (p^{(t)}) + Ku_{p^{(t)}}(p^{(t)})$\;
       }
       \caption{Best response dynamics Algorithm}
       \label{algo:best_response_dynamics}
\end{algorithm}
\end{minipage}
\end{center}

To investigate the effectiveness of this approach, more formally a Moran process
will be considered. If a population of \(n\) total individuals of two types is
considered, with \(K\) individuals of the first type and \(n-K\) of the second
type.
The probability that the individuals of the first type will take over the
population (the fixation probability) is denoted by \(x_K\) and is known to
be~\cite{nowak2006evolutionary}:

\[ x_K = \frac{ 1 + \sum_{j=1}^{K-1}\prod_{i=1}^j\gamma_i }{ 1 +
\sum_{j=1}^{n-1}\prod_{i=1}^j\gamma_i } \]

where:

\[ \gamma_i = \frac{ p_{K, K - 1} }{ p_{K, K + 1} }. \]

To evaluate the formulation proposed here the best response player (taken to be
the first type of individual in our population)  will be allowed to act
dynamically: adjusting their probability vector at every generation. In essence
using the theory of mind to find the best response to not only the opponent but
also the distribution of the population. Thus for every value of \(K\) there is
a different best response player.

Considering the dynamic best response player as a vector
\(p\in\mathbb{R}^4_{[0, 1]}\) and the opponent as a vector
\(q\in\mathbb{R}^4_{[0, 1]}\), the transition probabilities depend on
the payoff matrix \(A ^ {(K)}\) where:

\begin{itemize}
    \item \(A ^ {(K)}_{11}=u_{p}(p)\) is the long run utility of the best response player against itself.
    \item \(A ^ {(K)}_{12}=u_{q}(p)\) is the long run utility of the best response player against the opponent.
    \item \(A ^ {(K)}_{11}=u_{p}(q)\) is the long run utility of the opponent against the best response player.
    \item \(A ^ {(K)}_{11}=u_{q}(q)\) is the long run utility of the opponent against itself.
\end{itemize}

The matrix \(A ^ {(K)}\) is calculated using Eq.~\ref{eq:optimisation_quadratic}.
For every value of \(K\) the best response dynamics algorithm
(Algorithm~\ref{algo:best_response_dynamics})
is used to
calculate the best response player.

The total utilities/fitnesses for each player can be written down:

\[f_1^{(K)} = (K - 1) A_{11}^{(K)} + (n - K)A_{12}^{(K)}\]

\[f_2^{(K)} = (K) A_{21}^{(K)} + (n - K - 1)A_{22}^{(K)}\]

where \(f_1^{(K)}\) is the fitness of the best response player, and \(f_2^{(K)}\) is
the fitness of the opponent.

Using this:

\[ p_{K, K - 1} = \frac{ (n - K)f_2^{(K)} }{ Kf_1^{(K)}+(n - K)f_2^{(K)} } \frac{ K }{ n } \]

and:

\[ p_{K, K + 1} = \frac{ Kf_1^{(K)} }{ Kf_1^{(K)}+(n - K)f_2^{(K)} } \frac{ (n - K) }{ n } \]

which are all that are required to calculate \(x_K\).

Figure~\ref{fig:dynamic_moran_process_results} shows the results of an analysis
of \(x_K\) for dynamically updating players.
This is obtained over 182
Moran process against 122 randomly selected opponents.
For each Moran process the fixation probabilities for \(K\in\{1, 2, 3\}\) are
collected.
As well as recording \(x_K\), \(\tilde x_K\) is measured where \(\tilde
x_K\) represents the fixation probability of the best response player
calculated for a given \(K\) but not allowing it to dynamically update as the
population changes. The ratio \(\frac{x_K}{\tilde x_K}\) is included in the
Figure. This is done to be able to compare to a high performing strategy that
has a theory of mind of the opponent but not of the population density.
The ratio shows a relatively high performance compared to a non
dynamic best response strategy. The mean ratio over all values of \(K\) and all
experiments is \(1.044\).
In some cases this dynamic updating results in a 25\% increase in the
absorption probability.

As denoted before it is clear that the best response strategy in general does not have a
low SSE (only 25\% of the data is below .923 and the average is .454) this is
further compounded by the ratio being above one showing that in many
cases the dynamic strategy benefits from its ability to adapt.
This indicates that memory-one strategies that perform well in Moran processes need to
more adaptable than a ZDs, and aim for mutual cooperation as well as
exploitation which is in line with the results of~\cite{Hilbe2018} where their
strategy was designed to adapt and was shown to be evolutionary stable. The
findings of this work show that an optimal strategy acts in the same way.

\begin{figure}[!hbtp]
    \centering
    \includegraphics[width=.9\textwidth]{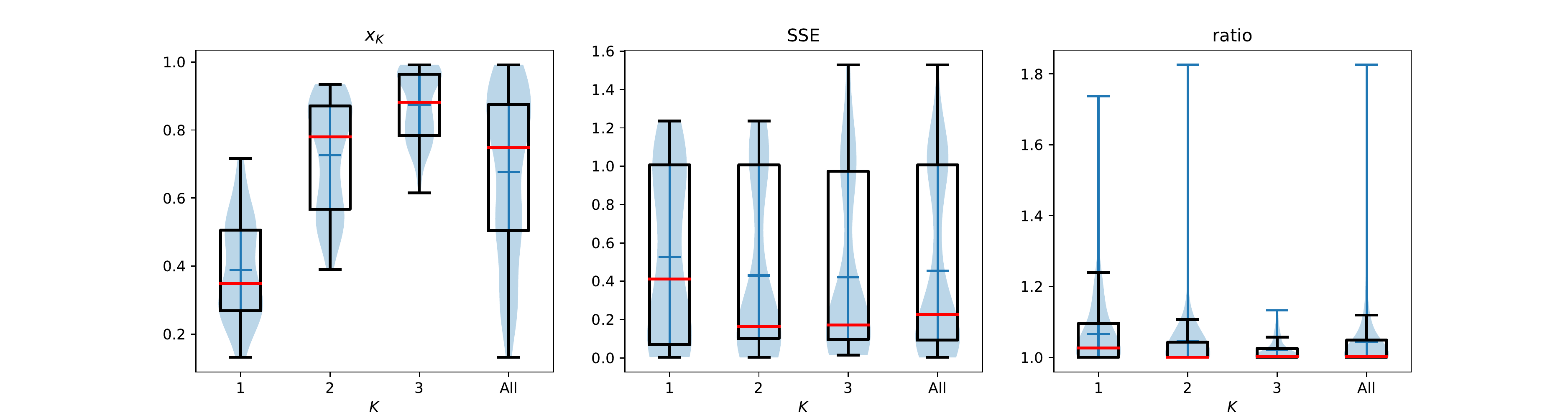}
    \caption{Results for the best response player in a dynamic Moran process.
    The ratio is taken as the ratio of \(x_k\) of the dynamically updating
    player to the fixation probability of a best response player that does not
    update as the population density changes.}\label{fig:dynamic_moran_process_results}
\end{figure}

\section{Discussion}

This manuscript has considered \textit{best response} strategies in the IPD game, and
more specifically, \textit{memory-one best responses}. It has proven that:

\begin{itemize}
    \item The utility
          of a memory-one strategy against a set of memory-one opponents can be written as a sum
          of ratios of quadratic forms (Theorem~\ref{theorem_one}).
    \item There is a compact way of identifying a memory-one best response to a
        group of opponents through a search over a discrete set
        (Theorem~\ref{memone_group_best_response}).
\end{itemize}

There is one further theoretical result that can be obtained from Theorem~\ref{theorem_one},
which allows the identification of
environments for which cooperation cannot occur (Details are in the
Appendix\ref{appendix:stability_of_defection}). Moreover,
Theorem~\ref{memone_group_best_response} does not only have game theoretic
novelty, but also the mathematical novelty of solving quadratic ratio
optimisation problems where the quadratics are non concave.

The empirical results of the manuscript investigated the behaviour of memory-one
strategies and their limitations. The empirical results have shown that the
performance of memory-one strategies rely on adaptability and not on extortion,
and that memory-one strategies' performance is limited by their memory in cases
where they interact with multiple opponents. These relied on two bespoke
data sets of 1000 and 152 pairs of memory-one opponents equivalently, archived
at~\cite{glynatsi2019}.

A further set of results for Moran processes with a dynamically updating best
response player was generated and is archived in~\cite{nikoleta_glynatsi_2020}.
This confirmed the previous results which is that high performance
requires adaptability and not extortion. It also provides a framework for future
stability of optimal behaviour in evolutionary settings.

In the interactions we have considered here the players do not make mistakes; their
actions were executed with perfect accuracy. Mistakes, however, are relevant in
the reasearch of repeated games~\cite{Boyd1989, Imhof2007, Nowak1993, Wu1995}. In
future work we would consider interactions with ``noise''. Noise can be
incoroporated into our formulation and it can be shown that the utility remains a ratio
of quadratic forms (Details see Appendix~\ref{appendix:noise}).
Another avenue of investigation would be to understand if and/or when an
evolutionary trajectory leads to a best response strategy.

By specifically exploring the entire space of memory-one strategies to identify
the best strategy for a variety of situations, this work adds to the literature
casting doubt
on the effectiveness of ZDs, highlights the importance of adaptability and provides
a framework for the continued understanding of these important questions.

\section{Acknowledgements}

A variety of software libraries have been used in this work:

\begin{itemize}
    \item The Axelrod library for IPD simulations~\cite{axelrodproject}.
    \item The Scikit-optimize library for an implementation of Bayesian optimisation~\cite{tim_head_2018_1207017}.
    \item The Matplotlib library for visualisation~\cite{hunter2007matplotlib}.
    \item The SymPy library for symbolic mathematics~\cite{sympy}.
    \item The Numpy library for data manipulation~\cite{walt2011numpy}.
\end{itemize}

\bibliographystyle{plain}
\bibliography{bibliography.bib}

\section{Appendix}

\subsection{Theorem~\ref{theorem_one} Proof}\label{appendix:theorem_one}

The utility of a memory one player \(p\) against an opponent \(q\), \(u_q(p)\),
can be written as a ratio of two quadratic forms on \(R^{4}_[0, 1]\).

\begin{proof}

    It was discussed that \(u_q(p)\) it is the product of the steady state
    vector \(v\) and the PD payoffs,

    \[u_q(p) = v \cdot (R, S, T, P).\]

    The steady state vector which is the solution to \(vM = v\) is given by
    \begingroup
    \tiny
    \begin{equation*}
    \begin{split}
        v =  & \left[ \frac{p_{2} p_{3} (q_{2} q_{4} - q_{3} q_{4}) + p_{2} p_{4} (q_{2} q_{3} - q_{2} q_{4} - q_{3} + q_{4}) +
        p_{3} p_{4} (- q_{2} q_{3} + q_{3} q_{4}) - p_{3} q_{2} q_{4} + p_{4}q_{4} (q_{2} - 1)}{\bar{v}} \right., \\
        & \left. \frac{p_{1} p_{3} (q_{1} q_{4} - q_{2} q_{4}) + p_{1} p_{4} (- q_{1} q_{2} + q_{1} + q_{2} q_{4} -
        q_{4}) + p_{3} p_{4} (q_{1} q_{2} - q_{1} q_{4} - q_{2} + q_{4}) + p_{3}q_{4} (q_{2} - 1) -
         p_{4} q_{2} (q_{4} + 1) + p_{4} (q_{4} - 1)}{\bar{v}} \right., \\
        & \left. \frac{- p_{1} p_{2} (q_{1} q_{4} - q_{3} q_{4}) - p_{1} p_{4} (- q_{1} q_{3} + q_{3} q_{4})
          + p_{1} q_{1} q_{4} - p_{2} p_{4} (q_{1} q_{3} - q_{1} q_{4} - q_{3} + q_{4}) -
          p_{2} q_{4} (q_{3}  + 1) - p_{4}q_{4} (q_{1} + q_{3}) - p_{4} (q_{3}
          + q_{4}) - q_{4}}{\bar{v}} \right., \\ 
        & \left. \frac{p_{1} p_{2} (q_{1} q_{2} - q_{1} - q_{2} q_{3} + q_{3}) + p_{1} p_{3} (- q_{1} q_{3} + q_{2} q_{3})
         - p_{1} q_{1} (q_{2} + 1) + p_{2} p_{3} (- q_{1} q_{2} + q_{1} q_{3}
         + q_{2} - q_{3}) + p_{2} (q_{3}q_{2}  - q_{2} - q_{3} - 1) +
          p_{3} (q_{1} q_{2} - q_{3}q_{2} - q_{2} - q_{3}) + q_{2} - 1}{\bar{v}}\right],
    \end{split}
    \end{equation*}
    \endgroup

    where,
    \begingroup
    \footnotesize
    \begin{equation*}
        \begin{split}
           \bar{v} = & \quad p_{1} p_{2} (q_{1} q_{2} - q_{1} q_{4} - q_{1} - q_{2} q_{3} + q_{3} q_{4} + q_{3}) - p_{1} p_{3} (q_{1} q_{3} - q_{1} q_{4} - q_{2} q_{3} + q_{2} q_{4}) -
           p_{1} p_{4} (q_{1} q_{2} - q_{1} q_{3} - q_{1} - q_{2} q_{4} + q_{3} q_{4} + q_{4}) - \\
           & \quad p_{1} q_{1} (q_{2} + q_{4} + 1) + p_{2} p_{3} (- q_{1} q_{2} + q_{1} q_{3} + q_{2} q_{4} + q_{2} - q_{3} q_{4} - q_{3})
           + p_{2} p_{4} (- q_{1} q_{3} + q_{1} q_{4} + q_{2} q_{3} - q_{2} q_{4}) + p_{2} q_{2} (q_{3} - 1) - p_{2} q_{3} (q_{4} - 1) + \\
           & \quad p_{2} (q_{4} + 1) +  p_{3} p_{4} (q_{1} q_{2} - q_{1} q_{4} - q_{2} q_{3} - q_{2} + q_{3} q_{4} + q_{4}) + p_{3} q_{2} q_{1} ( - p_{3} - 1) + p_{3} (q_{3} -
           q_{4}) - p_{4} (q_{1} q_{4} + q_{2} + q_{3} q_{4} - q_{3} + q_{4} - 1) + \\
           & \quad q_{2} - q_{4} - 1
        \end{split}
        \end{equation*}
    \endgroup

    The dot product of \(v \cdot (R, S, T, P)\) gives,

    \begingroup
    \scriptsize
    \begin{equation*}
    \begin{split}
        u_q(p) = & \frac{R \left(p_{2} p_{3} (q_{2} q_{4} - q_{3} q_{4}) + p_{2} p_{4} (q_{2} q_{3} - q_{2} q_{4} - q_{3} + q_{4}) +
        p_{3} p_{4} (- q_{2} q_{3} + q_{3} q_{4}) - p_{3} q_{2} q_{4} + p_{4}q_{4} (q_{2} - 1)\right)}{\bar{v}}  +  \\
        & \frac{S \left(p_{1} p_{3} (q_{1} q_{4} - q_{2} q_{4}) + p_{1} p_{4} (- q_{1} q_{2} + q_{1} + q_{2} q_{4} -
        q_{4}) + p_{3} p_{4} (q_{1} q_{2} - q_{1} q_{4} - q_{2} + q_{4}) + p_{3}q_{4} (q_{2} - 1) -
         p_{4} q_{2} (q_{4} + 1) + p_{4} (q_{4} - 1)\right)}{\bar{v}} + \\
        & \frac{T \left(- p_{1} p_{2} (q_{1} q_{4} - q_{3} q_{4}) - p_{1} p_{4} (- q_{1} q_{3} + q_{3} q_{4})
          + p_{1} q_{1} q_{4} - p_{2} p_{4} (q_{1} q_{3} - q_{1} q_{4} - q_{3} + q_{4}) -
          p_{2} q_{4} (q_{3}  + 1) - p_{4}q_{4} (q_{1} + q_{3}) - p_{4} (q_{3}
          + q_{4}) - q_{4}\right)}{\bar{v}} + \\ 
        & \frac{P \left(p_{1} (p_{2} (q_{1} q_{2} - q_{1} - q_{2} q_{3} + q_{3}) + p_{3} (- q_{1} q_{3} + q_{2} q_{3})
        - q_{1} (q_{2} + 1)) + p_{2} p_{3} ((- q_{1} q_{2} + q_{1} q_{3}
        + q_{2} - q_{3}) + (q_{3}q_{2}  - q_{2} - q_{3} - 1))\right)}{\bar{v}} + \\
        & \frac{P \left(p_{3} (q_{1} q_{2} - q_{3}q_{2} - q_{2} - q_{3}) + q_{2} - 1\right)}{\bar{v}} \implies \\
    \end{split}
    \end{equation*}
    \endgroup

    \begingroup
    \scriptsize
    \begin{equation*}
        u_q(p) =
        \left(
          \frac
            {\parbox{6in}{$ - p_{1} p_{2} (q_{1} - q_{3}) (P q_{2} - P - T q_{4}) + p_{1} p_{3} (q_{1} - q_{2}) (P q_{3} - S q_{4}) + p_{1} p_{4} (q_{1} - q_{4}) (S q_{2} - S - T q_{3}) + p_{2} p_{3} (q_{2} - q_{3}) (P q_{1} - P - R q_{4}) - $ \\
            $ p_{2} p_{4} (q_{3} - q_{4}) (R q_{2} - R - T q_{1} + T) + p_{3} p_{4} (q_{2} - q_{4}) (R q_{3} - S q_{1} + S) + p_{1} q_{1} (P q_{2} - P - T q_{4}) - p_{2} (q_{3} - 1) (P q_{2} - P - T q_{4}) + $ \\
            $ p_{3} (- P q_{1} q_{2} + P q_{2} q_{3} + P q_{2} - P q_{3} + R q_{2} q_{4} - S q_{2} q_{4} + S q_{4}) + p_{4} (- R q_{2} q_{4} + R q_{4} + S q_{2} q_{4} - S q_{2} - S q_{4} + S + T q_{1} q_{4} - T q_{3} q_{4} + T q_{3} - T q_{4}) $ \\
            \hspace*{6.7cm} $- P q_{2} + P + T q_{4}$
            }}
            {\parbox{6in}{$
            p_{1} p_{2} (q_{1} q_{2} - q_{1} q_{4} - q_{1} - q_{2} q_{3} + q_{3} q_{4} + q_{3}) + p_{1} p_{3} (- q_{1} q_{3} + q_{1} q_{4} + q_{2} q_{3} - q_{2} q_{4}) + p_{1} p_{4} (- q_{1} q_{2} + q_{1} q_{3} + q_{1} + q_{2} q_{4} - q_{3} q_{4} - q_{4}) +$ \\
            $ p_{2} p_{3} (- q_{1} q_{2} + q_{1} q_{3} + q_{2} q_{4} + q_{2} - q_{3} q_{4} - q_{3}) + p_{2} p_{4} (- q_{1} q_{3} + q_{1} q_{4} + q_{2} q_{3} - q_{2} q_{4}) + p_{3} p_{4} (q_{1} q_{2} - q_{1} q_{4} - q_{2} q_{3} - q_{2} + q_{3} q_{4} + q_{4}) + $ \\
            $ p_{1} (- q_{1} q_{2} + q_{1} q_{4} + q_{1}) + p_{2} (q_{2} q_{3} - q_{2} - q_{3} q_{4} - q_{3} + q_{4} + 1) + p_{3} (q_{1} q_{2} - q_{2} q_{3} - q_{2} + q_{3} - q_{4}) + p_{4} (- q_{1} q_{4} + q_{2} + q_{3} q_{4} - q_{3} + q_{4} - 1) + $ \\
            \hspace*{7cm} $q_{2} - q_{4} - 1$
          }}
        \right).
    \end{equation*}
    \endgroup

    Let us consider the numerator of \(u_q(p)\). The cross product terms
    \(p_ip_j\) are given by,

    \begingroup
    \footnotesize
    \begin{align*}
    - p_{1} p_{2} (q_{1} - q_{3}) (P q_{2} - P - T q_{4}) + p_{1} p_{3} (q_{1} - q_{2}) (P q_{3} - S q_{4}) + p_{1} p_{4} (q_{1} - q_{4}) (S q_{2} - S - T q_{3}) + \\
    p_{2} p_{3} (q_{2} - q_{3}) (P q_{1} - P - R q_{4}) - p_{2} p_{4} (q_{3} - q_{4}) (R q_{2} - R - T q_{1} + T) + p_{3} p_{4} (q_{2} - q_{4}) (R q_{3} - S q_{1} + S)
    \end{align*}
    \endgroup

    This can be re written in a matrix format given by
    Eq.~\ref{eq:cross_product_coeffs}.

    \begin{equation}\label{eq:cross_product_coeffs}
        \resizebox{0.9\linewidth}{!}{\arraycolsep=2.5pt%
        \boldmath\(
        (p_1, p_2, p_3, p_4) \frac{1}{2} \input{tex/q_numerator} \begin{pmatrix}
        p_1 \\
        p_2 \\
        p_3 \\
        p_4 \end{pmatrix}
        \) }
    \end{equation}

    Similarly, the linear terms are given by,

    \begingroup
    \footnotesize
    \begin{align*}
    p_{1} q_{1} (P q_{2} - P - T q_{4}) - p_{2} & (q_{3} - 1) (P q_{2} - P - T q_{4}) + p_{3} (- P q_{1} q_{2} + P q_{2} q_{3} + P q_{2} - P q_{3} + R q_{2} q_{4} - S q_{2} q_{4} + S q_{4}) + \\
    p_{4} & (- R q_{2} q_{4} + R q_{4} + S q_{2} q_{4} - S q_{2} - S q_{4} + S + T q_{1} q_{4} - T q_{3} q_{4} + T q_{3} - T q_{4})
    \end{align*}
    \endgroup

    and the expression can be written using a matrix format as
    Eq.~\ref{eq:linear_coeffs}.

    \begin{equation}\label{eq:linear_coeffs}
        \resizebox{0.60\linewidth}{!}{\arraycolsep=2.5pt%
        \boldmath\(
        (p_1, p_2, p_3, p_4) \input{tex/c_numerator}\)}
    \end{equation}

    Finally, the constant term of the numerator, which is obtained by
    substituting $p=(0, 0, 0, 0)$, is given by Eq.~\ref{eq:constant}.

    \begin{equation}\label{eq:constant}
    - P q_{2} + P + T q_{4}
    \end{equation}

    Combining Eq.~\ref{eq:cross_product_coeffs}, Eq.~\ref{eq:linear_coeffs} and
    Eq.~\ref{eq:constant} gives that the numerator of \(u_q(p)\) can be written
    as,

    \begingroup
    \tiny\boldmath
    \begin{align*}
        \frac{1}{2}p & \input{tex/q_numerator} p^T +  \\
        & \input{tex/c_numerator} p - P q_{2} + P + T q_{4}
    \end{align*}
    \endgroup

    and equivalently as,

    \[\frac{1}{2}pQp^T + cp + a\]

    where \(Q\) \(\in \R^{4\times4}\) is a square matrix defined by the
    transition probabilities of the opponent \(q_1, q_2, q_3, q_4\) as follows:

    \begin{equation*}
        \resizebox{0.9\linewidth}{!}{\arraycolsep=2.5pt%
        \boldmath\(
        Q = \input{tex/q_numerator}\)},
    \end{equation*}

    \(c\) \(\in \R^{4 \times 1}\) is similarly defined by:

    \begin{equation*}
        \resizebox{0.55\linewidth}{!}{\arraycolsep=2.5pt%
        \boldmath\(c = \input{tex/c_numerator}\),}
    \end{equation*}

    and \(a = \input{tex/numerator_constant}\).

    The same process is done for the denominator.
\end{proof}

\subsection{Theorem~\ref{memone_group_best_response} Proof}\label{appendix:memone_group_best_response}
The optimal behaviour of a memory-one strategy player \(p^* \in \R_{[0, 1]} ^
4\) against a set of \(N\) opponents \(\{q^{(1)}, q^{(2)}, \dots, q^{(N)} \}\)
for \(q^{(i)} \in \R_{[0, 1]} ^ 4\) is given by:

\[p^* = \textnormal{argmax}\sum\limits_{i=1} ^ N  u_q(p), \ p \in S_q.\]

The set \(S_q\) is defined as all the possible combinations of:

\begin{equation}\label{eq:s_q_set}
    S_q =
    \left\{p \in \mathbb{R} ^ 4 \left|
        \begin{aligned}
            \bullet\quad p_j \in \{0, 1\} & \quad \text{and} \quad \frac{d}{dp_k}
            \sum\limits_{i=1} ^ N  u_q^{(i)}(p) = 0 \\
            & \quad \text{for all} \quad j \in J \quad \&  \quad k \in K  \quad \text{for all} \quad J, K \\
            & \quad \text{where} \quad J \cap K = \O \quad
            \text{and} \quad J \cup K = \{1, 2, 3, 4\}.\\
            \bullet\quad  p \in \{0, 1\} ^ 4
        \end{aligned}\right.
    \right\}.
\end{equation}

\begin{proof}
    The optimisation problem of Eq.~\ref{eq:mo_tournament_optimisation}

    \begin{equation}\label{eq:mo_tournament_optimisation}
        \begin{aligned}
        \max_p: & \ \sum_{i=1} ^ {N} {u_q}^{(i)} (p)
        \\
        \text{such that}: & \ p \in \R_{[0, 1]}
        \end{aligned}
    \end{equation}

    can be written as:

    \begin{equation}\label{eq:mo_tournament_optimisation_standard}
        \begin{aligned}
        \max_p: & \ \sum_{i=1} ^ {N} {u_q}^{(i)} (p)
        \\
        \text{such that}: p_i & \leq 1 \text{ for } \in \{1, 2, 3, 4\} \\
        - p_i & \leq 0 \text{ for } \in \{1, 2, 3, 4\} \\
        \end{aligned}
    \end{equation}

    The optimisation problem has two inequality constraints and regarding the
    optimality this means that:

    \begin{itemize}
        \item either the optimum is away from the boundary of the optimization
        domain, and so the constraints plays no role;
        \item or the optimum is on the constraint boundary.
    \end{itemize}

    Thus, the following three cases must be considered:

    \textbf{Case 1:} The solution is on the boundary and any of the possible
    combinations for $p_i \in \{0, 1\}$ for $i \in \{1, 2, 3, 4\}$ are candidate
    optimal solutions.

    \textbf{Case 2:} The optimum is away from the boundary of the optimization
    domain and the interior solution $p^*$ necessarily satisfies the condition
    \(\frac{d}{dp} \sum\limits_{i=1} ^ N  u_q(p^*) = 0\).

    \textbf{Case 3:} The optimum is away from the boundary of the optimization
    domain but some constraints are equalities. The candidate solutions in this
    case are any combinations of $p_j \in \{0, 1\} \quad \text{and} \quad
    \frac{d}{dp_k} \sum\limits_{i=1} ^ N  u_q^{(i)}(p) = 0$ forall $ j \in J
    \text{ \& } k \in K \text{ forall } J, K \text{ where } J \cap K = \O
    \text{ and } J \cup K = \{1, 2, 3, 4\}.$

    Combining cases 1-3 a set of candidate solutions, denoted as \(S_q\), is
    constructed as: {\scriptsize
    \begin{equation*}
        S_q =
        \left\{p \in \mathbb{R} ^ 4 \left|
            \begin{aligned}
                \bullet\quad p_j \in \{0, 1\} & \quad \text{and} \quad \frac{d}{dp_k}
                \sum\limits_{i=1} ^ N  u_q^{(i)}(p) = 0
                \quad \text{for all} \quad j \in J \quad \&  \quad k \in K  \quad \text{for all} \quad J, K \\
                & \quad \text{where} \quad J \cap K = \O \quad
                \text{and} \quad J \cup K = \{1, 2, 3, 4\}.\\
                \bullet\quad  p \in \{0, 1\} ^ 4
            \end{aligned}\right.
        \right\}.
    \end{equation*}}

    The derivative of \(\sum\limits_{i=1} ^ N  u_q^{(i)}(p)\) calculated using
    the following property (see~\cite{Abadir2005} for details):

    \begin{equation}\label{eq:first_derivative_property}
    \frac{d x A x^T}{dx} =  2Ax.
    \end{equation}

    Using property~(\ref{eq:first_derivative_property}):

    \begin{equation}\label{eq:quadratics_derivatives}
    \frac{d}{dp} \frac{1}{2}pQp^T + cp + a = pQ + c \text{ and } \frac{d}{dp} \frac{1}{2}p\bar{Q}p^T + \bar{c}p + \bar{a} = p\bar{Q} + \bar{c}.
    \end{equation}

    Note that the derivative of \(cp\) is \(c\) and the constant disappears.
    Combining these it can be proven that:

    \begingroup
    \footnotesize
    \begin{align*}
    \frac{d}{dp} \sum\limits_{i=1} ^ N  u_q^{(i)}(p) & = \sum\limits_{i=1} ^ N \frac{\frac{d}{dp}(\frac{1}{2}pQ^{(i)}p^T + c^{(i)}p + a^{(i)} )(\frac{1}{2}p\bar{Q^{(i)}}p^T + \bar{c^{(i)}}p + \bar{a^{(i)}}) -
    \frac{d}{dp}(\frac{1}{2}p\bar{Q^{(i)}}p^T + \bar{c^{(i)}}p + \bar{a^{(i)}})(\frac{1}{2}pQ^{(i)}p^T + c^{(i)}p + a^{(i)})}{(\frac{1}{2}p\bar{Q^{(i)}}p^T + \bar{c^{(i)}}p + \bar{a^{(i)}})^2} \\
    & = \sum\limits_{i=1} ^ N \frac{(pQ^{(i)} + c^{(i)} +)(\frac{1}{2}p\bar{Q^{(i)}}p^T + \bar{c^{(i)}}p + \bar{a^{(i)}})}{(\frac{1}{2}p\bar{Q^{(i)}}p^T + \bar{c^{(i)}}p + \bar{a^{(i)}})^2} -
     \frac{(p\bar{Q^{(i)}}+ \bar{c^{(i)}})(\frac{1}{2}pQ^{(i)}p^T + c^{(i)}p + a^{(i)})}{(\frac{1}{2}p\bar{Q^{(i)}}p^T + \bar{c^{(i)}}p + \bar{a^{(i)}})^2}
    \end{align*}
    \endgroup

    For \(\frac{d}{dp} \sum\limits_{i=1} ^ N  u_q(p)\) to equal zero then:

    {\scriptsize
    \begin{align}\label{eq:polynomials_roots}
        \displaystyle\sum\limits_{i=1} ^ {N}
        \left(pQ^{(i)} + c^{(i)}\right) \left(\frac{1}{2} p\bar{Q}^{(i)} p^T + \bar{c}^{(i)} p + \bar{a}^ {(i)}\right)
        - \left(p\bar{Q}^{(i)} + \bar{c}^{(i)}\right) \left(\frac{1}{2} pQ^{(i)} p^T + c^{(i)} p + a^ {(i)}\right)
        & = 0, \quad {while} \\
        \displaystyle\sum\limits_{i=1} ^ {N} \frac{1}{2} p\bar{Q}^{(i)} p^T + \bar{c}^{(i)} p + \bar{a}^ {(i)} & \neq 0.
    \end{align}}

    The optimal solution to Eq.~\ref{eq:mo_tournament_optimisation} is the point
    from $S_q$ for which the utility is maximised.
\end{proof}

\subsection{Stability of defection}\label{appendix:stability_of_defection}

An additional theoretical result that is possible to obtain due to
Theorem~\ref{theorem_one}, is a condition for which in an
environment of memory-one opponents defection is the stable choice, based only
on the coefficients of the opponents.

This result is obtained by evaluating the sign of Eq.
\ref{eq:tournament_utility}'s derivative at \(p=(0, 0, 0, 0)\). If at that
point the derivative is negative, then the utility of a player only decreases if
they were to change their behaviour, and thus defection at that point is stable.

\begin{lemma}\label{lemma:stability_of_defection}
    In a tournament of \(N\) players \(\{q^{(1)}, q^{(2)}, \dots, q^{(N)} \}\)
    for \(q^{(i)} \in \R_{[0, 1]} ^ 4\)
    defection is stable if the transition probabilities of the
    opponents satisfy conditions Eq. \ref{eq:defection_condition_one} and Eq. \ref{eq:defection_condition_two}.

    \begin{equation}\label{eq:defection_condition_one}
        \sum_{i=1} ^ N (c^{(i)T} \bar{a}^{(i)} - \bar{c}^{(i)T} a^{(i)}) \leq 0
    \end{equation}

    while,

    \begin{equation}\label{eq:defection_condition_two}
        \sum_{i=1} ^ N \bar{a}^{(i)} \neq 0
    \end{equation}
\end{lemma}

\begin{proof}
    For defection to be stable the derivative of the utility
    at the point \(p = (0, 0, 0, 0)\) must be negative.

    Substituting \(p = (0, 0, 0, 0)\) in,

    \begin{align}\label{eq:mo_tournament_derivative}
        \frac{d}{dp} \sum\limits_{i=1} ^ {N} {u_q}^{(i)} (p) & = \displaystyle\sum\limits_{i=1} ^ {N}
        \frac{\left(pQ^{(i)} + c^{(i)}\right) \left(\frac{1}{2} p\bar{Q}^{(i)} p^T + \bar{c}^{(i)} p + \bar{a}^ {(i)}\right)}
        {\left(\frac{1}{2} p\bar{Q}^{(i)} p^T + \bar{c}^{(i)} p + \bar{a}^ {(i)}\right)^ 2}
        - \frac{\left(p\bar{Q}^{(i)} + \bar{c}^{(i)}\right) \left(\frac{1}{2} pQ^{(i)} p^T + c^{(i)} p + a^ {(i)}\right)}
        {\left(\frac{1}{2} p\bar{Q}^{(i)} p^T + \bar{c}^{(i)} p + \bar{a}^ {(i)}\right)^ 2}
    \end{align}

    gives:

    \begin{equation}
        \left.\frac{d\sum\limits_{i=1} ^ {N} {u_q}^{(i)} (p)}{dp} \right\rvert_{p=(0,0,0,0)} =
    \sum_{i=1} ^ N \frac{(c^{(i)} \bar{a}^{(i)} - \bar{c}^{(i)} a^{(i)})}
    {(\bar{a}^{(i)})^2}
    \end{equation}

    The sign of the numerator \( \displaystyle\sum_{i=1} ^ N (c^{(i)} \bar{a}^{(i)} - \bar{c}^{(i)} a^{(i)})\)
    can vary based on the transition probabilities of the opponents.
    The denominator can not be negative, and otherwise is always positive.
    Thus the sign of the derivative is negative if and only if
    \( \displaystyle\sum_{i=1} ^ N (c^{(i)} \bar{a}^{(i)} - \bar{c}^{(i)} a^{(i)}) \leq 0\).
\end{proof}

A numerical simulation demonstrating the result is given in Fig.~\ref{fig:stability_of_defection}.

\begin{figure}[!htbp]
    \centering
    \includegraphics[width=.4\linewidth]{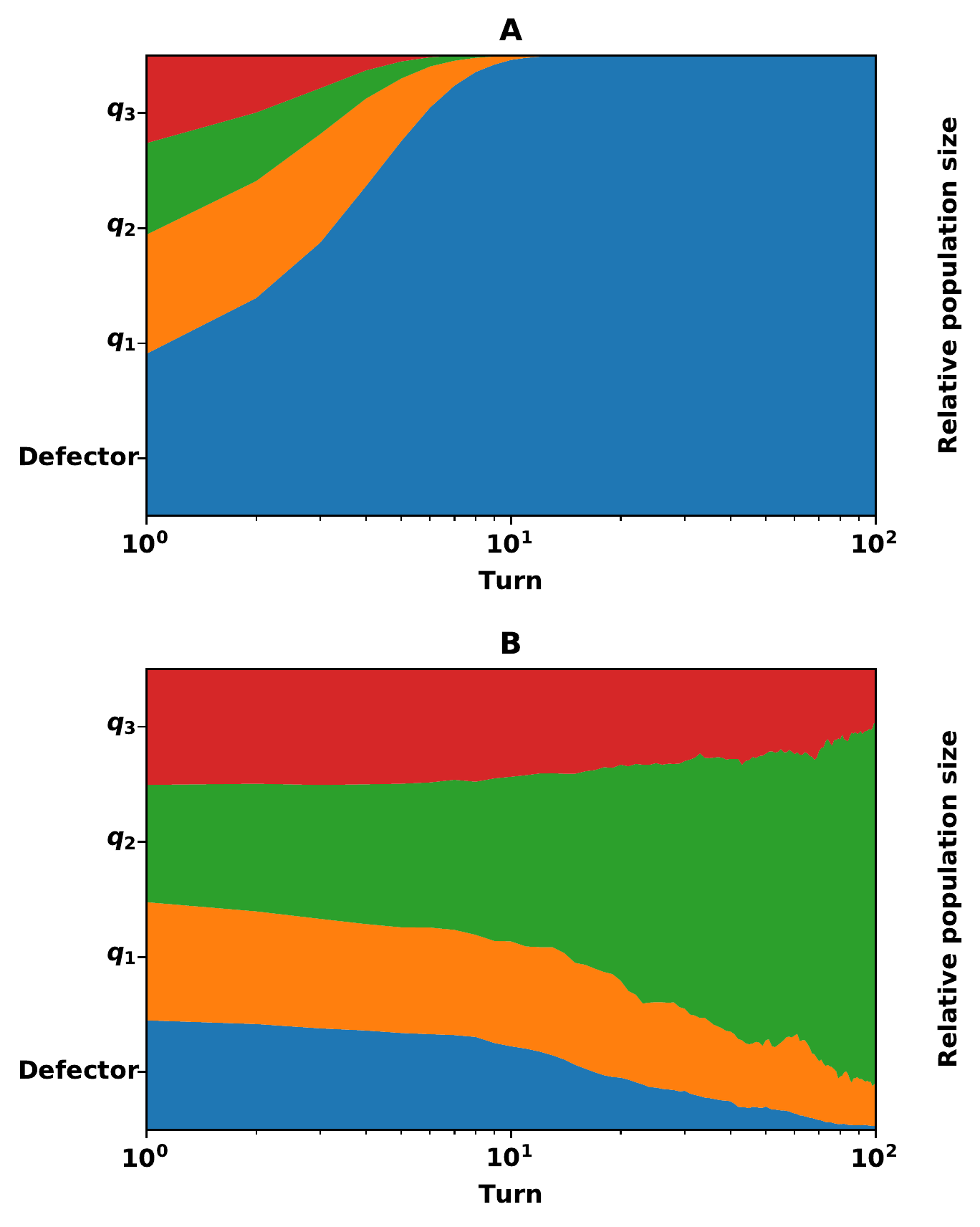}
    \caption{A. For \(q_{1}=(0.22199, 0.87073, 0.20672, 0.91861)\),
    $q_{2}=(0.48841, 0.61174, 0.76591, 0.51842)$ and
    $q_{3}=(0.2968, 0.18772, 0.08074, 0.73844)$, Eq.~\ref{eq:defection_condition_one} and
    Eq.~\ref{eq:defection_condition_two} hold and Defector takes over the
    population. B. For $q_{1}=(0.96703, 0.54723, 0.97268, 0.71482)$,
    $q_{2}=(0.69773, 0.21609, 0.97627, 0.0062)$ and
    $q_{3}=(0.25298, 0.43479, 0.77938, 0.19769)$, Eq.~\ref{eq:defection_condition_one} fails
    and Defector does not take over the population.
    These results have been obtained by using~\cite{axelrodproject} an open
    source research framework for the study of the IPD.}\label{fig:stability_of_defection}
\end{figure}

\subsection{Best response memory-one strategy in environments with noise}\label{appendix:noise}

Consider an environment where there is a probability \(p_n\) that a players actions are
executed wrong. This is referred to as the probability of noise.
Two memory-one opponents \(p \in \R_{[0, 1]} ^ 4\) and \(q \in \R_{[0, 1]} ^ 4\)
are now given by:

\[p = (p_1 (1 - p_n), p_2 (1 - p_n), p_3 (1 - p_n), p_4 (1 - p_n))\]

and

\[q = (q_1 (1 - p_n), q_2 (1 - p_n), q_3 (1 - p_n), q_4 (1 - p_n)).\]

Following a similar approach to that of Theorem~\ref{appendix:theorem_one} it
can be shown that the utility \(u_q(p)\) is give by:

\begin{equation}\label{eq:utility_with_noise}
    u_q(p) = \frac{\frac{1}{2}pQp^T + cp + a}
                {\frac{1}{2}p\bar{Q}p^T + \bar{c}p + \bar{a}},
\end{equation}
where \(Q, \bar{Q}\) \(\in \R^{4\times4}\) are square matrices defined by the
transition probabilities of the opponent \(q_1, q_2, q_3, q_4\) as follows:

\begin{center}
\begin{equation*}
\resizebox{\linewidth}{!}{\arraycolsep=2.5pt%
\boldmath\(
Q = \input{tex/q_numerator_with_noise}\)},
\end{equation*}
\begin{equation*}\label{eq:q_bar_matrix}
\resizebox{\linewidth}{!}{\arraycolsep=2.5pt%
\boldmath\(
\bar{Q} =  \input{tex/q_denominator_with_noise}\)}.
\end{equation*}
\end{center}

\(c \text{ and } \bar{c}\) \(\in \R^{4 \times 1}\) are similarly defined by:

\begin{equation}\label{eq:q_matrix_numerator}
\resizebox{0.7\linewidth}{!}{\arraycolsep=2.5pt%
\boldmath\(c = \input{tex/c_numerator_with_noise}\),}
\end{equation}
\begin{equation}\label{eq:q_matrix_denominator}
\resizebox{0.4\linewidth}{!}{\arraycolsep=2.5pt%
\boldmath\(\bar{c} = \input{tex/c_denominator_with_noise}\),
}
\end{equation}
and the constant terms \(a, \bar{a}\) are defined as \(a = \input{tex/numerator_constant_with_noise}\) and
\(\bar{a} = \input{tex/denominator_constant_with_noise}\).

\end{document}

%% file: tex/m_matrix.tex
M = \left[\begin{matrix}p_{1} q_{1} & p_{1} \left(- q_{1} + 1\right) & q_{1} \left(- p_{1} + 1\right) & \left(- p_{1} + 1\right) \left(- q_{1} + 1\right)\\p_{2} q_{3} & p_{2} \left(- q_{3} + 1\right) & q_{3} \left(- p_{2} + 1\right) & \left(- p_{2} + 1\right) \left(- q_{3} + 1\right)\\p_{3} q_{2} & p_{3} \left(- q_{2} + 1\right) & q_{2} \left(- p_{3} + 1\right) & \left(- p_{3} + 1\right) \left(- q_{2} + 1\right)\\p_{4} q_{4} & p_{4} \left(- q_{4} + 1\right) & q_{4} \left(- p_{4} + 1\right) & \left(- p_{4} + 1\right) \left(- q_{4} + 1\right)\end{matrix}\right]

%% file: tex/q_numerator.tex
\left[\begin{matrix}0 & - \left(q_{1} - q_{3}\right) \left(P q_{2} - P - T q_{4}\right) & \left(q_{1} - q_{2}\right) \left(P q_{3} - S q_{4}\right) & \left(q_{1} - q_{4}\right) \left(S q_{2} - S - T q_{3}\right)\\- \left(q_{1} - q_{3}\right) \left(P q_{2} - P - T q_{4}\right) & 0 & \left(q_{2} - q_{3}\right) \left(P q_{1} - P - R q_{4}\right) & - \left(q_{3} - q_{4}\right) \left(R q_{2} - R - T q_{1} + T\right)\\\left(q_{1} - q_{2}\right) \left(P q_{3} - S q_{4}\right) & \left(q_{2} - q_{3}\right) \left(P q_{1} - P - R q_{4}\right) & 0 & \left(q_{2} - q_{4}\right) \left(R q_{3} - S q_{1} + S\right)\\\left(q_{1} - q_{4}\right) \left(S q_{2} - S - T q_{3}\right) & - \left(q_{3} - q_{4}\right) \left(R q_{2} - R - T q_{1} + T\right) & \left(q_{2} - q_{4}\right) \left(R q_{3} - S q_{1} + S\right) & 0\end{matrix}\right]

%% file: tex/q_denominator.tex
\left[\begin{matrix}0 & - \left(q_{1} - q_{3}\right) \left(q_{2} - q_{4} - 1\right) & \left(q_{1} - q_{2}\right) \left(q_{3} - q_{4}\right) & \left(q_{1} - q_{4}\right) \left(q_{2} - q_{3} - 1\right)\\- \left(q_{1} - q_{3}\right) \left(q_{2} - q_{4} - 1\right) & 0 & \left(q_{2} - q_{3}\right) \left(q_{1} - q_{4} - 1\right) & \left(q_{1} - q_{2}\right) \left(q_{3} - q_{4}\right)\\\left(q_{1} - q_{2}\right) \left(q_{3} - q_{4}\right) & \left(q_{2} - q_{3}\right) \left(q_{1} - q_{4} - 1\right) & 0 & - \left(q_{2} - q_{4}\right) \left(q_{1} - q_{3} - 1\right)\\\left(q_{1} - q_{4}\right) \left(q_{2} - q_{3} - 1\right) & \left(q_{1} - q_{2}\right) \left(q_{3} - q_{4}\right) & - \left(q_{2} - q_{4}\right) \left(q_{1} - q_{3} - 1\right) & 0\end{matrix}\right]

%% file: tex/c_numerator.tex
\left[\begin{matrix}q_{1} \left(P q_{2} - P - T q_{4}\right)\\- \left(q_{3} - 1\right) \left(P q_{2} - P - T q_{4}\right)\\- P q_{1} q_{2} + P q_{2} q_{3} + P q_{2} - P q_{3} + R q_{2} q_{4} - S q_{2} q_{4} + S q_{4}\\- R q_{2} q_{4} + R q_{4} + S q_{2} q_{4} - S q_{2} - S q_{4} + S + T q_{1} q_{4} - T q_{3} q_{4} + T q_{3} - T q_{4}\end{matrix}\right]

%% file: tex/c_denominator.tex
\left[\begin{matrix}q_{1} \left(q_{2} - q_{4} - 1\right)\\- \left(q_{3} - 1\right) \left(q_{2} - q_{4} - 1\right)\\- q_{1} q_{2} + q_{2} q_{3} + q_{2} - q_{3} + q_{4}\\q_{1} q_{4} - q_{2} - q_{3} q_{4} + q_{3} - q_{4} + 1\end{matrix}\right]

%% file: tex/numerator_constant.tex
- P q_{2} + P + T q_{4}

%% file: tex/denominator_constant.tex
- q_{2} + q_{4} + 1

%% file: tex/q_numerator_with_noise.tex
\left[\begin{matrix}0 & - p_{n}^{3} \left(q_{1} - q_{3}\right) \left(P p_{n} q_{2} - P - T p_{n} q_{4}\right) & p_{n}^{4} \left(q_{1} - q_{2}\right) \left(P q_{3} - S q_{4}\right) & p_{n}^{3} \left(q_{1} - q_{4}\right) \left(S p_{n} q_{2} - S - T p_{n} q_{3}\right)\\- p_{n}^{3} \left(q_{1} - q_{3}\right) \left(P p_{n} q_{2} - P - T p_{n} q_{4}\right) & 0 & p_{n}^{3} \left(q_{2} - q_{3}\right) \left(P p_{n} q_{1} - P - R p_{n} q_{4}\right) & - p_{n}^{3} \left(q_{3} - q_{4}\right) \left(R p_{n} q_{2} - R - T p_{n} q_{1} + T\right)\\p_{n}^{4} \left(q_{1} - q_{2}\right) \left(P q_{3} - S q_{4}\right) & p_{n}^{3} \left(q_{2} - q_{3}\right) \left(P p_{n} q_{1} - P - R p_{n} q_{4}\right) & 0 & p_{n}^{3} \left(q_{2} - q_{4}\right) \left(R p_{n} q_{3} - S p_{n} q_{1} + S\right)\\p_{n}^{3} \left(q_{1} - q_{4}\right) \left(S p_{n} q_{2} - S - T p_{n} q_{3}\right) & - p_{n}^{3} \left(q_{3} - q_{4}\right) \left(R p_{n} q_{2} - R - T p_{n} q_{1} + T\right) & p_{n}^{3} \left(q_{2} - q_{4}\right) \left(R p_{n} q_{3} - S p_{n} q_{1} + S\right) & 0\end{matrix}\right]

%% file: tex/q_denominator_with_noise.tex
\left[\begin{matrix}0 & - p_{n}^{3} \left(q_{1} - q_{3}\right) \left(p_{n} q_{2} - p_{n} q_{4} - 1\right) & p_{n}^{4} \left(q_{1} - q_{2}\right) \left(q_{3} - q_{4}\right) & p_{n}^{3} \left(q_{1} - q_{4}\right) \left(p_{n} q_{2} - p_{n} q_{3} - 1\right)\\- p_{n}^{3} \left(q_{1} - q_{3}\right) \left(p_{n} q_{2} - p_{n} q_{4} - 1\right) & 0 & p_{n}^{3} \left(q_{2} - q_{3}\right) \left(p_{n} q_{1} - p_{n} q_{4} - 1\right) & p_{n}^{4} \left(q_{1} - q_{2}\right) \left(q_{3} - q_{4}\right)\\p_{n}^{4} \left(q_{1} - q_{2}\right) \left(q_{3} - q_{4}\right) & p_{n}^{3} \left(q_{2} - q_{3}\right) \left(p_{n} q_{1} - p_{n} q_{4} - 1\right) & 0 & - p_{n}^{3} \left(q_{2} - q_{4}\right) \left(p_{n} q_{1} - p_{n} q_{3} - 1\right)\\p_{n}^{3} \left(q_{1} - q_{4}\right) \left(p_{n} q_{2} - p_{n} q_{3} - 1\right) & p_{n}^{4} \left(q_{1} - q_{2}\right) \left(q_{3} - q_{4}\right) & - p_{n}^{3} \left(q_{2} - q_{4}\right) \left(p_{n} q_{1} - p_{n} q_{3} - 1\right) & 0\end{matrix}\right]

%% file: tex/c_numerator_with_noise.tex
\left[\begin{matrix}p_{n}^{2} q_{1} \left(P p_{n} q_{2} - P - T p_{n} q_{4}\right)\\- p_{n} \left(p_{n} q_{3} - 1\right) \left(P p_{n} q_{2} - P - T p_{n} q_{4}\right)\\- p_{n}^{2} \left(P p_{n} q_{1} q_{2} - P p_{n} q_{2} q_{3} - P q_{2} + P q_{3} - R p_{n} q_{2} q_{4} + S p_{n} q_{2} q_{4} - S q_{4}\right)\\- p_{n} \left(R p_{n}^{2} q_{2} q_{4} - R p_{n} q_{4} - S p_{n}^{2} q_{2} q_{4} + S p_{n} q_{2} + S p_{n} q_{4} - S - T p_{n}^{2} q_{1} q_{4} + T p_{n}^{2} q_{3} q_{4} - T p_{n} q_{3} + T p_{n} q_{4}\right)\end{matrix}\right]

%% file: tex/c_denominator_with_noise.tex
\left[\begin{matrix}p_{n}^{2} q_{1} \left(p_{n} q_{2} - p_{n} q_{4} - 1\right)\\- p_{n} \left(p_{n} q_{3} - 1\right) \left(p_{n} q_{2} - p_{n} q_{4} - 1\right)\\- p_{n}^{2} \left(p_{n} q_{1} q_{2} - p_{n} q_{2} q_{3} - q_{2} + q_{3} - q_{4}\right)\\p_{n} \left(p_{n}^{2} q_{1} q_{4} - p_{n}^{2} q_{3} q_{4} - p_{n} q_{2} + p_{n} q_{3} - p_{n} q_{4} + 1\right)\end{matrix}\right]

%% file: tex/numerator_constant_with_noise.tex
P + p_{n} \left(- P q_{2} + T q_{4}\right)

%% file: tex/denominator_constant_with_noise.tex
p_{n} \left(- q_{2} + q_{4}\right) + 1